\documentclass[12pt,letterpaper]{article}
\pdfoutput=1

\usepackage{geometry} 
\usepackage{microtype}
\usepackage{hyperref}
\hypersetup{pdfauthor={Boris Aronov and Micha Sharir},pdftitle={Almost Tight Bounds for Eliminating Depth Cycles in Three Dimensions}}
\usepackage{mathtools}
\usepackage{enumerate}
\usepackage{color,latexsym,amssymb,amsthm,graphicx,subfigure}
\usepackage{cite}

\usepackage[small]{titlesec}

\DeclareMathOperator{\polylog}{polylog}

\usepackage{fullpage}

\def\A{\mathcal{A}}

\def\L{\mathcal{L}}

\let\eps\varepsilon
\def\reals{\mathbb{R}}

\newtheorem{theorem}{Theorem}[section]
\newtheorem{lemma}[theorem]{Lemma}

\theoremstyle{remark}
\newtheorem*{remark}{Remark}

\begin{document}

\title{Almost Tight Bounds\\for Eliminating Depth Cycles\\in Three Dimensions\thanks{%
    Work on this paper by B.A.\ has been partially supported by NSF
    Grants CCF-11-17336, CCF-12-18791, and CCF-15-40656, and by BSF grant 2014/170.
    Work by M.S.\ has been supported by Grant 2012/229 from the
    U.S.-Israel Binational Science Foundation, by Grant 892/13 from
    the Israel Science Foundation, by the Israeli Centers for Research
    Excellence (I-CORE) program (center no.~4/11), and by the Hermann
    Minkowski--MINERVA Center for Geometry at Tel Aviv University.
    A preliminary version of this paper has appeared in {\it Proc.\ 48th Annu.\ ACM Symposium on Theory of Computing}, 2016\cite{stoc}.}
}

\author{%
  Boris Aronov\thanks{%
    Department of Computer Science and Engineering, Tandon School of Engineering,
    New York University, Brooklyn, NY~11201, USA;
    \textsl{boris.aronov@nyu.edu}.}%
  \and
  Micha Sharir\thanks{%
    Blavatnik School of Computer Science, Tel Aviv University, Tel-Aviv 69978,
    Israel; \textsl{michas@post.tau.ac.il}.} }

\maketitle

\begin{abstract}
  Given $n$ non-vertical lines in 3-space, their vertical depth
  (above/below) relation can contain cycles. We show that the lines
  can be cut into $O(n^{3/2}\polylog n)$ pieces, such that the depth
  relation among these pieces is now a proper partial order.  This
  bound is nearly tight in the worst case.  

  Previous results on this topic could only handle restricted cases of
  the problem (such as handling only triangular cycles, by Aronov,
  Koltun, and Sharir (2005), or only cycles in grid-like patterns, by
  Chazelle et al.~(1992)), and the bounds were considerably weaker---much
  closer to the trivial quadratic bound.

  Our proof uses a recent variant of the polynomial partitioning
  technique, due to Guth, and some simple tools from algebraic
  geometry. It is much more straightforward than the previous ``purely
  combinatorial'' methods.

  Our technique can be extended to eliminating all cycles in the depth
  relation among segments, and among constant-degree algebraic arcs.  
  We hope that a suitable extension of this technique could be used to handle
  the much more difficult case of pairwise-disjoint triangles.

  We also discuss several algorithms for constructing a small set of
  cuts so as to eliminate all depth-relation cycles among the lines
  (minimizing such a set, for the case of line segments, is known to
  be NP-complete).  The performance of these algorithms improves due
  to our new bound, but, so far, none of them both (a)~produce close
  to $n^{3/2}$ cuts, and (b)~run in time close to $n^{3/2}$, in the
  worst case.

  Our results almost completely settle a 35-year-old
  open problem in computational geometry, motivated by hidden-surface
  removal in computer graphics.
\end{abstract}

\section{Introduction}
\label{sec:intro}

\paragraph*{The problem.}
Let $\L$ be a collection of $n$ non-vertical lines in $\reals^3$ in
\emph{general position}. In particular, we assume that no two lines in $\L$
intersect, that the $xy$-projections of no pair of the lines are
parallel, and those of no three of the lines are concurrent.
For any pair $\ell,\ell'$ of lines in $\L$, we say that
$\ell$~passes \emph{above} $\ell'$ (equivalently, $\ell'$ passes
\emph{below} $\ell$) if the unique vertical line that meets
both $\ell$ and $\ell'$ intersects $\ell$ at a point that lies higher
than its intersection with~$\ell'$. We denote this relation as
$\ell'\prec\ell$. The relation $\prec$ is total (under our assumptions),
but in general it
need not be transitive, so it may contain \emph{cycles} of the form
$\ell_1\prec\ell_2\prec\cdots\prec\ell_k\prec\ell_1$, for some $k\ge 3$. We call this a
\emph{$k$-cycle}, and refer to $k$ as the \emph{length} of the cycle.
Cycles of length three are called \emph{triangular}. See Figure~\ref{ccwcyc}.
\begin{figure}[hbpt]
  \centering
  \includegraphics{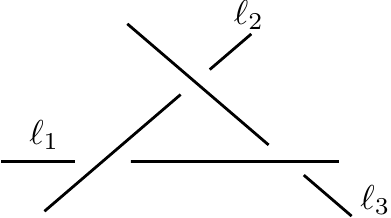}
  \caption{A triangular depth cycle, viewed from above.}
  \label{ccwcyc}
\end{figure}

If we cut the lines of $\L$ at a finite number of points, we obtain a
collection of lines, segments, and rays. We can extend the definition
of the relation $\prec$ to the new collection in the obvious manner,
except that it is now only a partial relation. Our goal is to cut the
lines in such a way that $\prec$ becomes a \emph{partial order}, in which
case we call it a \emph{depth order}. We note that it is trivial to
construct a depth order with $\Theta(n^2)$ cuts: Simply cut each line
near every point whose $xy$-projection is a crossing point with another
projected line. It is desirable though to minimize the number of cuts.
A long-standing conjecture, open since~1980, is that one can always
construct a depth order with a \emph{subquadratic} number of cuts.
In this paper we finally settle this conjecture, in a strong, almost worst-case tight
manner; see below for precise details.

\paragraph*{Background.}
The main motivation for studying this problem comes from
\emph{hidden surface removal} in computer graphics.
A detailed description of this motivation can be found, e.g., in the
earlier paper of Aronov et al.~\cite{AKS}.
Briefly, a conceptually simple technique for rendering a scene in
computer graphics is the so-called Painter's Algorithm, which renders
the (pairwise openly disjoint) objects in the scene on the screen in a back-to-front manner,
painting each new object over the portions of earlier objects that it hides.
For this, though, one needs an acyclic depth relation among the objects with
respect to the viewing point (which, as we assume in this paper, without loss of generality,
lies at $z=+\infty$). When there are cycles in the depth relation, one would
like to cut the objects into a small number of pieces, so as to eliminate
all cycles, and then paint the pieces in the above manner, obtaining a correct
rendering of the scene; see \cite{AKS,4M-book} for more details.

The study of depth cycles in a set of lines (or, more generally, of segments) in $\reals^3$ 
was in a rather poor state circa 1990, where only very weak combinatorial and algorithmic bounds were available. The situation is reviewed in de Berg's 1992 dissertation \cite[Chapter 9]{MdB-thesis} (where he summarizes the state of affairs before the results described in his thesis and those of Chazelle et al.~\cite{CEG+} were obtained; see the discussion below). At that time it was not even known how to compute a depth order among $n$ line segments in 3-space in subquadratic time, assuming that such an order exists, or how to verify in subquadratic time that a given order of $n$ segments is acyclic.

The study of subquadratic bounds for depth cycles goes back to
Chazelle et al.~\cite{CEG+}, who have shown that, if the $xy$-projections
of a collection of $n$ segments in 3-space form a ``grid'' (see Figure~\ref{grid} for an illustration 
and the paper for a formal definition),
\begin{figure}[hpbt]
  \centering
  \includegraphics{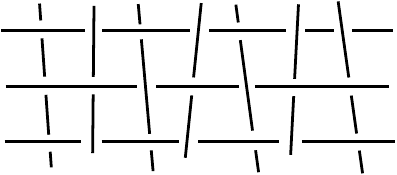}
  \caption{A collection of line segments that forms a grid, viewed from below.}
  \label{grid}
\end{figure}
then all cycles defined by this collection can be eliminated with $O(n^{9/5})$ cuts.
Another significant development is due to Aronov et al.~\cite{AKS},
who have considered the problem of triangular cycles,
and established the rather weak (albeit subquadratic)
$O(n^{2-1/34}\log^{8/17}n)$ upper bound on the number of \emph{elementary triangular} cycles
(namely, cycles whose $xy$-projections\footnote{%
  Here, and throughout the paper, we think of a cycle
  $C:\; \ell_1\prec\ell_2\prec\cdots\prec\ell_k\prec\ell_1$ as a spatial geometric
  object; it is the unique closed polygonal path whose edges alternate between segments along
  the lines and vertical jumps between consecutive lines that connect the unique pairs of 
  co-vertical points on those lines. This realization will be defined formally in Section~\ref{sec:cut}.}
form triangular faces in the arrangement of the projected lines).
They also showed that $O(n^{2-1/69}\log^{16/69}n)$ cuts suffice to eliminate
all triangular cycles. Finally, combining this bound with earlier algorithmic
techniques of Solan \cite{So} and of Har-Peled and Sharir \cite{HPS}, they
have obtained an algorithm that eliminates all triangular cycles by making
roughly $O(n^{2-1/138})$ cuts. However, their results did not apply to general
(non-triangular) cycles, and, in addition to the very weak bounds just stated,
the proof technique was very involved.  Just as the analysis of 
Chazelle et al.~\cite{CEG+}, it used, as a major ingredient, the
impossibility of certain ``weaving patterns'' of lines in space, an interesting
and intriguing topic in itself. Unfortunately, it appears that arguments
based on forbidden weaving patterns lead to fairly weak bounds.

In another, more recent development, Aronov et al.~\cite{ABGM} have shown that
finding the minimum number $\chi$ of cuts needed to eliminate all cycles in a collection
of segments in $\reals^3$ is NP-complete. They have also presented a deterministic 
polynomial-time approximation algorithm that constructs $O(\chi\log\chi\log\log\chi)$ cuts.
More details are given below.

\paragraph*{Our contribution.}
In this paper we settle the general problem and show that \emph{all} cycles
in a set of $n$ lines can be eliminated with $O(n^{3/2}\polylog n)$ cuts.
A simple and well-known construction, reviewed below, yields a scenario where
$\Omega(n^{3/2})$ cuts have to be made, implying that ours is the best possible
worst-case bound, up to the polylogarithmic factor.

The proof of the new bound is embarrassingly straightforward. It uses
tools from algebraic geometry, in the spirit of much recent work
that exploited similar ideas; see, e.g., the simple proofs in \cite{KSS,Qu}
for the corresponding worst-case tight bound of $\Theta(n^{3/2})$ on the number
of so-called \emph{joints} in a collection of $n$ lines in 3-space. At the heart of the construction
lies a recent result of Guth \cite{Gut}, which extends the basic
\emph{polynomial partitioning} technique of Guth and Katz \cite{GK2}
to higher-dimensional objects (lines or curves in our case).

As a matter of fact, the algebraic approach to this problem is fairly versatile and can be extended 
to the elimination of cycles involving more general objects. In this paper we also discuss such extensions
to the cases of line segments (this is in fact a trivial extension) and of constant-degree algebraic arcs.
Furthermore, in both cases, by combining our technique with standard tools for constructing output-sensitive 
cuttings in the plane, we obtain improved bounds on the number of cuts necessary to eliminate all cycles, 
which depend on the number of 
intersections among the $xy$-projections of the segments or arcs. See Theorems~\ref{thm:segments} and \ref{thm:arcs}
for further details.

We note that the practical motivation arising from computer graphics involves data sets consisting of 
pairwise openly-disjoint triangles. Eliminating cycles in the depth relation of a collection
of triangles (with a subquadratic number of cuts) is a considerably more difficult problem, which so far
has barely been touched. The case of lines, as studied in this paper and in the previous cited works,
is a simpler instance of the problem, which already turned out to be a very challenging question.
It is our hope that the technique presented in this paper could be extended to tackle this case too, and we
are presently investigating such an extension.

We also note that the problem studied here is different from most of the combinatorial geometry
questions tackled so far by the new algebraic approach,  in that these former problems
involve \emph{incidences} between points, lines, and other objects. In 
contrast, in this paper the lines are not incident to one another, and the configurations that
we want to capture involve certain non-contact spatial (here, ``above/below'') relationships between them.
It is our hope that this study will find applications to additional problems involving relations more general than incidences.

Our proof is constructive, and leads, in principle, to an efficient algorithm
for performing the cuts. The only currently missing ingredient is an
efficient construction of Guth's partitioning polynomial, a step that we
leave as a topic for further research.
(The problematic aspects of effectively constructing a partitioning
polynomial, already for the simpler case of a set of points, and techniques
for overcoming these issues are discussed by Agarwal et al.~\cite{AMS};
one hopes that an extension of these techniques could also be used for
effectively partitioning lines or curves, and we plan to investigate this question further.)

Alternatively, to identify the cuts sufficient to eliminate all
cycles, one could also use the earlier algorithms of
Har-Peled and Sharir \cite{HPS} or of Solan \cite{So}.  
Our analysis implies that they perform $O(n^{7/4}\polylog n)$ cuts
(using the algorithm of \cite{HPS}; the one in \cite{So} generates a slightly larger number of cuts),
in expected time $O(n^{11/6+\eps})$, for any $\eps>0$, significantly improving previous 
bounds, but still falling short of the ideal goal of achieving a running time close to the worst-case optimal number of cuts. 

Another approach is to use the deterministic approximation algorithm of 
Aronov et al.~\cite[Theorem 3.1]{ABGM}, mentioned earlier. 
It constructs $O(\chi\log\chi\log\log\chi)$ cuts that eliminate all cycles, 
where $\chi$ is the smallest number of such cuts, in $O(n^{4+2\omega}\log^2 n)$ 
time, where $\omega < 2.373$ is the exponent of matrix multiplication.
In view of our main result, this algorithm produces $O(n^{3/2}\polylog n)$ cuts
(where the power in the polylogarithmic factor is slightly larger than that in Theorem~\ref{thm:cuts}).
This is probably the best algorithm one can offer at the present state of affairs.
Its only drawback is its rather large running time (which is $O(n^{8.746})$); the
algorithms of \cite{HPS,So} are much faster, but they produce a larger number of cuts.

\paragraph*{Paper organization.}
Section~\ref{sec:cut} presents the main result on the number of cuts sufficient to 
eliminate all depth cycles among lines in $\reals^3$.
Section~\ref{sec:alg} discusses the algorithmic aspects of efficiently finding such a set of cuts. 
Finally, Section~\ref{sec:ext} discusses the extensions of our technique to 
the cases of line segments and of constant-degree algebraic arcs. 

\section{Eliminating all cycles}
\label{sec:cut}

We first introduce a few definitions.
Let $\L$ be a collection of $n$ non-vertical lines in $\reals^3$ in general position.
For each $\ell\in \L$, denote by $\ell^*$ the $xy$-projection of $\ell$ and by
$\L^*$ the collection of the $n$ resulting projections. The general position
assumption on $\L$ implies that $\L^*$ is also in general position, in the traditional sense, that is, with no two lines parallel and no three concurrent.
Consider the planar arrangement $\A(\L^*)$ of $\L^*$.

Recall that $k$ distinct lines $\ell_1,\ldots,\ell_k$, for $k\ge 3$, form a \emph{$k$-cycle}
$C$ if $\ell_1\prec \ell_2\prec \cdots \prec \ell_k\prec \ell_1$.
We can interpret $C$ as a spatial object as follows. For each $i=1,\ldots,k$,
let $v_i^+\in\ell_i$ and $v_{i+1}^-\in\ell_{i+1}$ (with indices treated $\mathop{\mathrm{mod}} k$)
be the two unique points on these lines that are vertically above each other (informally,
$C$ ``jumps upwards'' from $v_i^+$ on $\ell_i$ to $v_{i+1}^-$ on $\ell_{i+1}$).
We associate with $C$ the closed directed polygonal path
\[
\pi(C) := v_1^-v_1^+v_2^-v_2^+\cdots v_k^-v_k^+v_1^- .
\]
Let $e_i$ denote the segment $v_i^-v_i^+$ on $\ell_i$. Then $\pi(C)$ alternates
between the segments $e_i$ and the vertical jumps $v_i^+v_{i+1}^-$; see Figure~\ref{fig:cycle-and-projection}.
\begin{figure}
  \centering
  \includegraphics[scale=1.1]{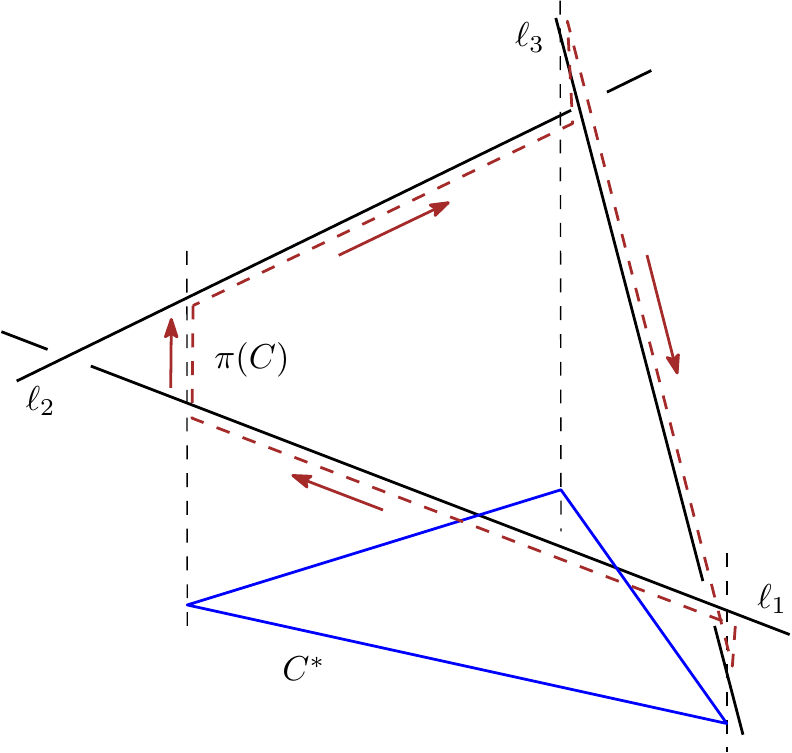}
  \caption{Cycle $C \colon \ell_1 \prec \ell_2 \prec \ell_3 \prec \ell_1$ (thick lines), 
  with the corresponding path $\pi(C)$ (in dashed brown), and its projection $C^*$ (solid blue).}
  \label{fig:cycle-and-projection}
\end{figure}

We say that $C$ is \emph{eliminated} if at least one of the non-vertical edges $v_i^-v_i^+$ of $\pi(C)$
is cut. It is an easy (yet crucial) observation that if we cut the lines of $\L$ at some discrete set of points,
so that all cycles in $\L$ are eliminated, in the above sense, then the collection of lines, rays, and segments
resulting from the cuts has an acyclic depth relation.

The $xy$-projection $C^*$ of $\pi(C)$ (or, with a slight abuse of notation, of $C$)
is a closed polygonal path contained in $\cup \L^*$.
That is, it is the concatenation of the projections $e_i^*$ of the segments $e_i$
(the vertical segments disappear, or rather shrink to points, in the projection).

The path $C^*$ can be fairly arbitrary, non-convex, and even self-crossing.
Nevertheless, we claim that, for the purposes of eliminating all cycles,
it suffices to consider only \emph{simple cycles},%
\footnote{%
  This is a slight abuse of terminology, as we require the projection $C^*$
  of the cycle $C$ to be simple, rather than $C$ itself.}
that is, cycles $C$ for which $C^*$ is non-self-crossing. This is because any
other cycle $C$ can be shortcut to a cycle $C_0$, such that (a)~$C_0$ has fewer
edges than $C$, and (b)~$C_0^*\subset C^*$. Clearly, any cut that eliminates $C_0$
also eliminates $C$. We repeat this reduction until we obtain a simple cycle
(in the extreme, we reach the case where $C_0$ is triangular, and thus simple).
Indeed, if $C^*$ is self-crossing, let $w$ be a point where $C^*$ crosses itself,
and let $\ell$, $\ell'$ be the lines whose projections cross at $w$; see Figure~\ref{fig:shortcut}.
Then $C^*$ is naturally split at $w$ into two shorter closed paths, and it is easily checked that one of them
is the projection of a cycle $C_0$ in $\L$ that satisfies the properties claimed above.

In what follows we thus restrict ourselves to simple cycles only.
Two (simple) cycles $C_1$, $C_2$ \emph{overlap} if $C_1^*$ and $C_2^*$
share an edge of $\A(\L^*)$; for example, in Figure~\ref{fig:shortcut}, cycles $C$ and $C_0$ overlap.

\begin{figure}[hpbt]
  \centering
  \includegraphics[width=0.6\textwidth]{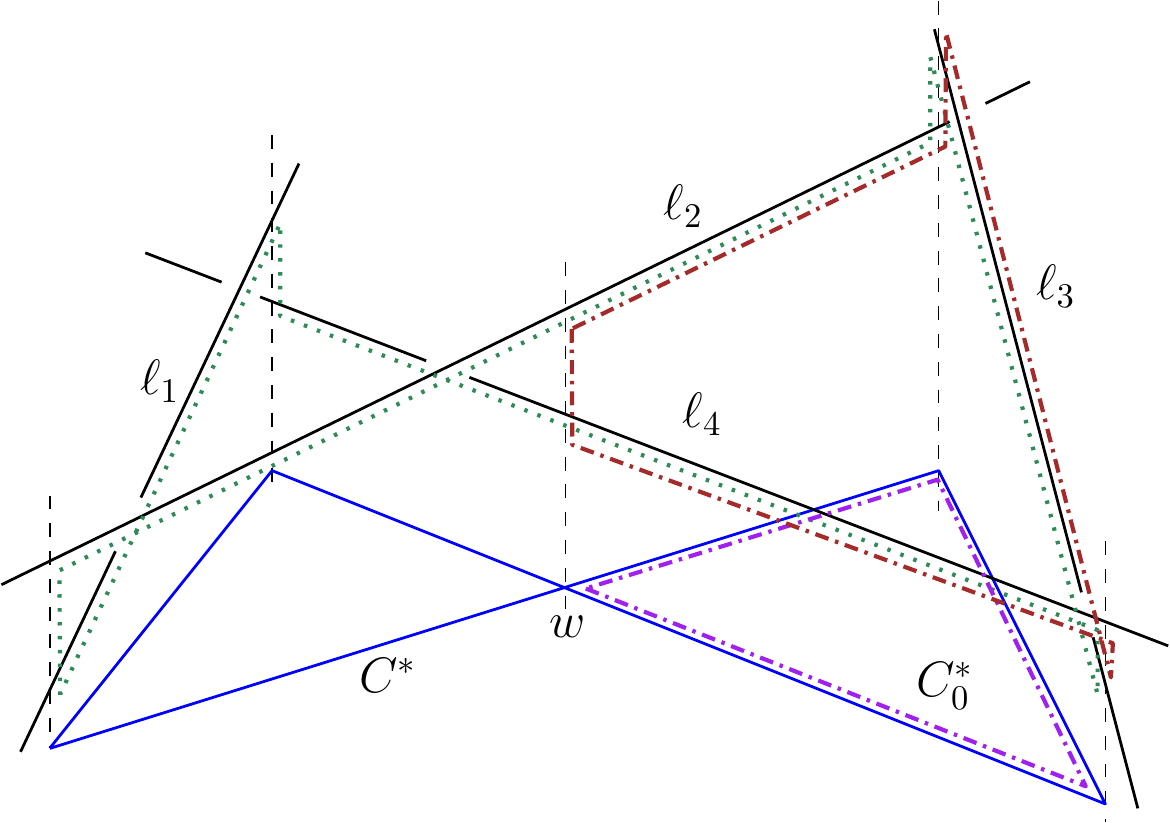}
  \caption{Cycle $C \colon \ell_1 \prec \ell_2 \prec \ell_3 \prec \ell_4 \prec \ell_1$ (thick lines, 
  with corresponding path $\pi(C)$ in dotted green), whose projection $C^*$ (solid blue) crosses itself.  
  We shortcut it to a new triangular cycle $C_0 \colon \ell_2 \prec \ell_3 \prec \ell_4 \prec \ell_2$, 
  indicating the corresponding path $\pi(C_0)$ in dash-dotted brown and its projection $C_0^*$ in dash-dotted purple.}
  \label{fig:shortcut}
\end{figure}

The following is the main result of the paper.
\begin{theorem} \label{thm:cuts}
Let $\L$ be a collection of $n$ non-vertical lines in $\reals^3$ in general position.
The lines of $\L$ can be cut at $O(n^{3/2}\polylog n)$ points so that the
depth relation on the resulting pieces (lines, rays, and segments) has no cycles.
This bound is almost tight in the worst case, since $\Omega(n^{3/2})$ cuts are sometimes necessary.
\end{theorem}

\begin{remark}
  Theorem~\ref{thm:cuts} also provides the same upper bound on the
  maximum size of any family~$F$ of pairwise non-overlapping cycles in
  $\L$, because a distinct cut is required to eliminate each cycle
  of~$F$.
\end{remark}

\begin{proof}
As argued above, it suffices to cut all simple cycles.
We fix some degree $D=D(n)$, which depends on $n$ and whose value will be set below,
and construct a non-zero trivariate polynomial $f\in\reals[x,y,z]$ of degree at most $D$, such
that each of the $s=O(D^3)$ open connected components (\emph{cells})
of $\reals^3\setminus Z(f)$ is
intersected by at most $cn/D^2$ lines of $\L$, where $Z(f)$ denotes the zero set of $f$,
and $c$ is an absolute constant.
By the aforementioned recent result of Guth \cite{Gut}, such a polynomial does exist,
for some suitable constant $c$.
(As already noted, effective and efficient construction of such a polynomial remains
to be worked out, and is the only reason this proof is not entirely polynomial-time
constructive.) Let $\tau_1,\ldots,\tau_s$ be the cells of $\reals^3\setminus Z(f)$,
and, for each $i$, let $\L_i$ denote the set of lines of $\L$ that intersect $\tau_i$.

In what follows we want to exclude situations where $Z(f)$ fully contains a vertical
segment (and therefore, a line).
We can guarantee that this does not happen, by applying a
sufficiently small generic ``tilting'' to the coordinate frame, ensuring that this property holds,
and that every simple cycle in $\L$ remains a (simple) cycle.

Define the \emph{level} $\lambda(q)$ of a 
point $q\in\reals^3$ with respect to $Z(f)$ to be the number of intersection points of $Z(f)$ with the
downward-directed open vertical ray $\rho_q$ emanating from~$q$. Formally, let
$(x_0,y_0,z_0)$ be the coordinates of $q$, and consider the univariate polynomial
$F(z) = f(x_0,y_0,z)$. The level $\lambda(q)$ of $q$ is the number of real zeros of $F$ in
$(-\infty,z_0)$, counted with multiplicity.

Denote by $\chi(\L)$ the minimum number of cuts needed to eliminate all (simple) cycles
in the given set $\L$ of lines, and put $\chi(n) := \max_{|\L|=n} \chi(\L)$, where the
maximum is taken over all collections $\L$ of $n$ non-vertical lines in general
position in $\reals^3$.

\paragraph*{The procedure for cutting the lines.}
The procedure is recursive and follows the partitioning induced by $Z(f)$.
It consists of the following steps.

\medskip

\noindent{\bf (i)}
We cut each line $\ell\in \L$ not fully contained in $Z(f)$
at all its intersection points with $Z(f)$. The number of such cuts is
at most $D$ per line, for a total of $O(nD)$ cuts.

\medskip

\noindent{\bf (ii)}
For each line $\ell\in \L$ not fully contained in $Z(f)$,
let $h(\ell)$ be the vertical plane containing~$\ell$, and
let $g_\ell$ be the bivariate polynomial obtained by restricting $f$ to $h(\ell)$.
More formally, parametrize $h(\ell)$ by orthogonal coordinates $(\xi_\ell,z)$, where $\xi_\ell$ is
horizontal, and each $(\xi_\ell,z)$ represents a point $(x_\ell(\xi_\ell),y_\ell(\xi_\ell),z)$ in $h(\ell)$,
where $x_\ell(\cdot)$, $y_\ell(\cdot)$ are appropriate linear functions depending on~$\ell$.
Then $g_\ell$ is given by $g_\ell(\xi_\ell,z) := f(x_\ell(\xi_\ell),y_\ell(\xi_\ell),z)$; it is a bivariate polynomial
of degree at most $D$. By removing repeated factors, we may assume that $g_\ell$ is
square-free. We then cut $\ell$, in addition to the cuts made in step~(i), at each
point that lies directly above a singular point, or a point of vertical tangency,
of $Z(g_\ell) \subset h(\ell)$.  A simple application of B\'ezout's theorem implies that the number of
such points is $O(D^2)$, because each such point is a common zero of $g_\ell$ and $\partial g_\ell /\partial z$.
Note that to apply B\'ezout's theorem, we need to ensure that $g_\ell$ and $\partial g_\ell /\partial z$
do not have a common factor, which is indeed the case since $g_\ell$ is square-free.
Hence we perform in this step $O(D^2)$ cuts of each line, for a total of $O(nD^2)$ cuts.

\medskip

\noindent{\bf (iii)}
Assume next that $\ell \subset Z(f)$; since $Z(f)$ contains no vertical lines, $h(\ell) \not\subset Z(f)$.
Let $g_\ell$ be the (square-free) bivariate polynomial defined in step (ii). 
Then $\ell \subset Z(g_\ell)$ is an irreducible component of $Z(g_\ell)$.
By removing the linear factor corresponding to~$\ell$, we replace $g_\ell$ by another
square-free polynomial $g^0_\ell$, of degree smaller than $D$, whose zero set does not fully
contain $\ell$. We then cut $\ell$ at each point where it meets $Z(g^0_\ell)$ (this
is a variant of step (i)), and at each point that lies directly above a critical
point of $g^0_\ell$, as defined earlier (a variant of step (ii)). 
As before, the number of such cuts of~$\ell$ is $O(D^2)$, for a total of $O(nD^2)$ cuts.

\medskip

\noindent{\bf (iv)}
We now proceed recursively: For each cell $\tau_i$ of the partition, we recurse on
the corresponding subset $\L_i$ of lines. The bottom of the recursion is at cells
$\tau_i$ for which $n_i:=|\L_i| < D^2(n_i)/c$. For such cells we apply the na\"ive procedure,
noted in the introduction, which cuts the lines of $\L_i$ into $O(n_i^2) = O(D^4(n_i))$
pieces, so that all cycles in $\L_i$ are trivially eliminated.

\begin{lemma} \label{cutok}
The procedure described above eliminates all the cycles in $\L$.
\end{lemma}
\begin{proof}
The proof is by induction on the size of the input. The claim holds
at the bottom of recursion, because we make all possible cuts there, thereby eliminating 
all cycles. Consider then any non-terminal instance of the recursion, involving
some subset of lines, which, for convenience, we again call $\L$.

As argued above, it suffices to show that we have cut all simple cycles.
Let $C$ be a simple cycle in $\L$, formed by some $k$ lines $\ell_1,\ldots,\ell_k$,
with $\ell_1\prec \ell_2\prec\cdots\prec \ell_k\prec \ell_1$, and $k\ge 3$. Let $C^*$ denote the
$xy$-projection of~$C$, which is a simple polygon with $k$ sides $e_1^*,\ldots,e_k^*$,
so that, for $i=1,\ldots,k$, $e_i^*\subset\ell_i^*$ is the $xy$-projection of the
corresponding segment $e_i\subset\ell_i$ of the path $\pi(C)$.

If $Z(f)$ does not intersect $\pi(C)$ then $\pi(C)$ is fully contained in some cell $\tau_i$,
and therefore all the lines $\ell_1,\ldots,\ell_k$ belong to $\L_i$, since they all intersect~$\tau_i$. 
By induction, the cycle~$C$ will be eliminated by the recursive call to the line-cutting
procedure on~$\L_i$. Assume then, in what follows, that $Z(f)$ intersects $\pi(C)$.

Assume first that $Z(f)$ does not fully contain any of the lines $\ell_1,\ldots,\ell_k$.
If $Z(f)$ intersects (but does not contain) one of the segments $e_i$, for $i=1,\ldots,k$,
then this intersection point, at which we have cut $\ell$ in step (i), eliminates the
cycle~$C$.

Assume next that none of the lines $\ell_1,\ldots,\ell_k$ is fully contained in $Z(f)$,
and that none of the segments $e_i\subset \ell_i$ is crossed by $Z(f)$.
In this case, the crossing points of $\pi(C)$ with $Z(f)$ must all lie on
the vertical edges of $\pi(C)$. Recall that we have ensured that $Z(f)$ does not fully
contain any such segment.

Trace $\pi(C)$ in a circular fashion, as in its definition,
and keep track of the level $\lambda(q)$ in $Z(f)$ of the point $q$ being traced.
By our general position assumption, and by the tilting performed above,
the level is well defined, and it can change only either (i) at the vertical jumps of 
$\pi(C)$, or (ii) at a discrete set of points $q$ on the edges $e_i$ of $\pi(C)$ 
at which the univariate restriction of $f$ to the vertical line through $q$ has a 
multiple real root (recall that, by assumption, $Z(f)$ does not intersect any $e_i$); 
see below for a discussion of this statement.
Each time we go up along one of the vertical segments of $\pi(C)$, the
level either increases or stays the same, and, by assumption, it strictly increases at
least once along the cycle, when $\pi(C)$ intersects $Z(f)$.  Since the levels at the beginning and at
the end of the tour are the same, the level must go down at least
once, as we trace one of the segments $e_i$, $i=1,\ldots,k$. Suppose,
without loss of generality, that the level goes down as we trace
$e_1$.  This must happen at a point $q \in\ell_1$ 
at which the univariate restriction of $f$ to the vertical line through $q$ has a 
multiple real root. That is, $q$ lies vertically above a point at which
$g_{\ell_1} = \partial g_{\ell_1}/\partial z = 0$, where $g_{\ell_1}$ refers here to the original
version of the restriction of $f$ to $h(\ell_1)$. Now if $g_{\ell_1}$ is square-free, we are done,
since, by construction, we have cut $\ell_1$ at $q$, and thus $C$ got eliminated
by this cut. If $g_{\ell_1}$ is not square-free, it is possible that
$g_{\ell_1} = \partial g_{\ell_1}/\partial z = 0$ along a one-dimensional curve, so this
property holds for a continuum of points $q$ on $\ell_1$. However, in such a case
(i) the multiplicity of the root does not cause the level to change at $q$, and
(ii) this vanishing on a one-dimensional curve does not occur for the square-free version
of $g_{\ell_1}$. This implies that the change in level must occur at a criticality of 
the square-free version of $g_{\ell_1}$, and $\ell_1$ has been cut above every such 
criticality, implying that $C$ has been eliminated in this case as well.

Finally, consider the case where one (or more) of the lines $\ell_1,\ldots,\ell_k$ is
fully contained in~$Z(f)$; say $\ell_1$ is such a line.
If one of the edges $e_i$ of $C$ has been cut by steps (i)--(iii),
we are done, so assume that this did not happen. As a point $q$ traces $\pi(C)$,
as above, the level~$\lambda(q)$ goes up at least once, when we go up from $v_1^+$ to
$v_2^-$ (at $v_2^-$ we count $\ell_1$ (i.e., $v_1^+\in Z(f)$) in its level, whereas at $v_1^+$ we do not).
Since the level cannot go down along any of the vertical upward edges of~$\pi(C)$,
it must go down when $q$ traverses some edge $e_i$ of $C$.  Therefore, arguing
as above, $q$ must lie directly above a critical point of the square-free version
of the restricted polynomial~$g_{\ell_i}$, or of its reduced version $g^0_{\ell_i}$
if $\ell_i \subset Z(f)$ ($\ell_1$ falls into this latter case), for some $i=1,\dots,k$.  
In either case, $\ell_i$ has been cut at $q$ and $C$ has been eliminated.

Having covered all cases, the lemma follows.
\end{proof}

It remains to bound the number of cuts. 
Collecting the bounds from each step of our construction and maximizing over $\L$
produces the recurrence 
\[
\chi(n) \le \begin{cases}
bD^3(n)\chi(cn/D^2(n)) + O(nD^2(n)) , & \text{for $n > D^2(n)/c ,$} \\
bD^4(n), & \text{for $n\le D^2(n)/c$} ,
\end{cases} 
\]
where $c$ is the absolute constant from Guth's construction, mentioned above, and $b$ is another suitable absolute constant.

Setting $D(n)=c^{1/2}n^{1/4}$, the termination condition $n\le D^2(n)/c$ becomes $n\le 1$, in which case
no cuts are needed.\footnote{%
  For the algorithmic part of the analysis, it is preferable to work with constant degree $D$,
  which is why, up to this point, we have not committed to a particular choice of~$D=D(n)$; 
  see the remark below and Section~\ref{sec:alg} for more details.}
That is, $\chi(1)=0$, and for $n>1$ we have 
\[
\chi(n) \le b_1n^{3/4}\chi(n^{1/2}) + O(n^{3/2}) ,
\]
for another absolute constant $b_1 = bc^{3/2}$.
Since the depth of recursion is $O(\log\log n)$, its solution is easily seen to be
$\chi(n) = O(n^{3/2}\log^{\beta}n )$, where $\beta$ is a constant that depends only on
the absolute constant $b_1$. This completes the proof of the upper bound.


\paragraph*{Lower bound.}
The near-tightness of the bound follows from the grid-like construction of $\Theta(n^{3/2})$
\emph{joints} (points incident to at least three non-coplanar lines)
in a collection of $n$ (or rather $3n$) lines, where the joints are the vertices of the
$\sqrt{n}\times \sqrt{n}\times \sqrt{n}$ integer grid, and the lines are the $3n$ axis-parallel
lines of the grid; see, e.g., \cite{AKS,CEG+}.
By slightly perturbing (translating and tilting) each of the lines, and by appropriately
tilting the coordinate frame, each joint is mapped to a small elementary triangular cycle in 
the arrangement of $\Theta(n)$ lines in general position in $\reals^3$.
As the cycles do not overlap, each requires a separate cut.
\end{proof}

\begin{remark}
  (1) Setting $D$ to a sufficiently large constant, rather than a function
  of $n$, in the above argument allows us to avoid having to work with
  polynomials whose degree depends on~$n$, at the expense of slightly weakening the
  upper bound to $O(n^{3/2+\eps})$, where $\eps=\eps(D)>0$ depends on
  the choice of $D$ and can be made arbitrarily small.  Of course, the
  implied constant in the big-Oh grows with~$D$.

\smallskip

\noindent
  (2) In the opposite direction, sharpening the bound further by increasing $D$ does not work in the
  present setting, because of step (ii) of the construction (and the corresponding portion of step (iii)),
  where each line is cut at $O(D^2)$ points. If only step (i) were sufficient, the non-recursive overhead would
  have been only $O(nD)$, and then we could choose $D=O(n^{1/2})$, skipping the recursion altogether, to obtain
  the worst-case optimal bound $O(n^{3/2})$. It is an interesting challenge to improve the bound along these lines.
\end{remark}

\section{Algorithmic considerations}
\label{sec:alg}

We now outline and discuss several algorithms for eliminating cycles.  Notice that identifying a smallest possible set of 
cuts to eliminate all cycles for a given family of line segments is shown in Aronov et al.~\cite{ABGM} to be NP-complete.

\paragraph*{Implementing the procedure in the proof of Theorem~\ref{thm:cuts}.}
The most straightforward way to obtain the required cuts would be to implement the 
mostly-constructive proof of Theorem~\ref{thm:cuts}, except that we set $D$ to be
a sufficiently large constant (see the remark after the proof), in order to control
the cost of the algebraic calculations that are needed to determine the cutting points.

However, this would require an effective and efficient construction of
the partitioning polynomial of Guth \cite{Gut}, the existence which is currently 
an open problem. One may hope that the techniques
developed in Agarwal et al.~\cite{AMS} for effective construction of
``approximate partitioning polynomials'' for sets of points would be
helpful here as well, especially since the degree $D$ is now assumed to be a constant. 
However, the machinery employed by Guth to
prove the existence of the said polynomial is sufficiently different
to make an extension to this case a serious challenge.

The rest of the algorithm would proceed as in the proof of Theorem~\ref{thm:cuts}.
One needs to assume a suitable model of algebraic computation that supports
constant-time execution of each of the various primitive algebraic operations
required by the algorithm (such as finding the intersections of a
line with $Z(f)$, finding the critical points of the polynomials $g_\ell$, etc.) 
for constant-degree polynomials. See, e.g., Basu et al.\cite{BPR} for a discussion
of the existing machinery for implementing operations of this kind.

\paragraph*{The algorithms of Har-Peled and Sharir and of Solan.}
Alternatively, we can use the algorithms of Har-Peled and Sharir\cite{HPS} or of Solan\cite{So},
specifically designed to eliminate cycles in the depth relation.  Given a collection~$\L$ of 
$n$~lines (or line segments) in $\reals^3$, these algorithms work on the arrangement $\A(\L^*)$ 
of the $xy$-projections of the lines, and partition the plane into regions, either by a cutting 
(in Solan \cite{So}), or by incrementally refining regions into subregions (in Har-Peled and Sharir \cite{HPS}),
exploiting the fact that one can efficiently detect the presence of a depth cycle in a collection of line segments
in $\reals^3$ using an algorithm of De~Berg et al.~\cite{dBOS}.
Both algorithms generate close to $O(n\sqrt{\chi})$ cuts, where $\chi$ is the minimum number 
of cuts required to eliminate all cycles in the given set of segments. 
Concretely, the slightly improved randomized algorithm in \cite{HPS}
makes $O(n\sqrt{\chi}\alpha(n)\log n)$ cuts in expectation (the bound in Solan's algorithm 
\cite{So} is slightly worse), and runs in expected time $O(n^{4/3+\eps}\chi^{1/3})$, 
for any $\eps>0$; see \cite[Theorem~6.1]{HPS} and \cite[Theorem~2.1]{So};
here $\alpha(\cdot)$ is the inverse Ackermann function. Therefore, we may conclude:
\begin{theorem}
  There exists a randomized algorithm that, given a set of $n$~lines in~$\reals^3$,
  can find a set of $O(n^{7/4}\polylog n)$ cuts sufficient to eliminate all cycles in 
  the depth relation among the lines, in expected time $O(n^{11/6+\eps})$, for any $\eps>0$.
\end{theorem}
Clearly, the algorithm falls short of
the ideal double goal of (a)~finding a set of cuts close to the minimum possible size (or just of size close to $n^{3/2}$), as in the
(not yet fully polynomial-time) construction in the proof of Theorem~\ref{thm:cuts}, and (b)~doing it in time close to $n^{3/2}$.
The main merit of this approach is its efficiency, as it runs in subquadratic time.

\paragraph*{The approximation algorithm of Aronov et al.}
The third algorithm, at the present state of affairs, appears to be the best of the three
approaches. It is due to Aronov et al.~\cite[Theorem 3.1]{ABGM}, and is based on an approximation
algorithm of Even et al.~\cite{ENRS} for Feedback Vertex Cover. It is a deterministic algorithm
that produces a set of $O(\chi\log\chi\log\log\chi)$ cuts that eliminate all cycles, where
$\chi$ is the smallest size of such a set, and runs in $O(n^{4+2\omega}\log^2 n) = O(n^{8.746})$
time, where $\omega < 2.373$ is the exponent of matrix multiplication. That is, this yields a
(not very efficient but still) polynomial-time algorithm that gets very close to the minimum 
number of cuts needed. In particular, it generates $O(n^{3/2}\polylog n)$ cuts (where the power 
in the polylogarithmic factor is slightly larger than the one in Theorem~\ref{thm:cuts}).

\section{Extensions to line segments and algebraic arcs}
\label{sec:ext}

In this section we discuss two extensions of our technique, to sets of line segments
and of constant-degree algebraic arcs.

\paragraph*{The case of line segments.}
Consider a non-degenerate set $S$ of $n$ non-vertical line segments in $\reals^3$,
and let $\A(S^*)$ be the arrangement formed by their $xy$-projections; as we assume 
general position, each vertex of $\A(S^*)$ is either the projection of an endpoint
of a segment in $S$, or a proper crossing of two projected segments. 
Let $X$~denote the number of vertices of the latter kind; we refer to them as \emph{proper} vertices.

Of course, suitably perturbed, $S$ can be extended to a set of lines
in general position, and therefore all cycles in $S$ can be eliminated
using $O(n^{3/2}\polylog n)$ cuts, by Theorem~\ref{thm:cuts}.  
We want to refine this bound, to make it depend on $X$.

Clearly the case $X=0$ requires no cuts, and if $X\le n$, we cut every segment $s$ 
near each point projecting to a proper vertex of $\A(S^*)$, thereby making $O(n)$ cuts and eliminating all cycles.

For larger values of $X$, set $r:= n^2/X < n$, and construct a
$(1/r)$-cutting of $S^*$ with $O(r+\frac{r^2}{n^2}X) = O(r)$ trapezoids, 
each crossed by at most $n/r$ segments in $S^*$ \cite{dBS}.
Cut every segment of $S$ at each point lying vertically above the boundary of a 
trapezoid of the cutting, thereby making $O(r)\cdot(n/r) = O(n)$ cuts.
Now apply the bound of Theorem~\ref{thm:cuts} over each trapezoid separately, concluding that
\[
O\left( n + r(n/r)^{3/2} \polylog(n/r) \right) =
O(n^{1/2}X^{1/2}\polylog n)
\]
cuts are sufficient to eliminate all cycles.
Combining the two cases and using the algorithms of Har-Peled and Sharir\cite{HPS} or of Solan\cite{So}, 
or of Aronov et al.~\cite{ABGM}, we conclude:
\begin{theorem}
  \label{thm:segments}
  The number of cuts sufficient to eliminate all cycles in a family of $n$ 
  non-vertical line segments 
  in general position in $\reals^3$ is $O(n + n^{1/2}X^{1/2}\polylog n)$, where $X$~is the 
  number of pairs of segments whose $xy$-projections cross.

One can compute (i)~$O(n+n^{3/4}X^{1/4}\polylog n)$ cuts that eliminate all cycles, in expected 
time $O(n^{3/2+\eps}X^{1/6})$, for any $\eps>0$, or
(ii)~$O(n\log n\log\log n + n^{1/2}X^{1/2}\polylog n)$ cuts that eliminate all cycles, in deterministic
time $O(n^{4+2\omega}\log^2n)$, where $\omega$ is the matrix multiplication exponent.
\end{theorem}

\paragraph*{The case of algebraic arcs.}
Our argument can be extended, with fairly minor adjustments, to a similar bound on
the number of cuts needed to eliminate all cycles in a collection of $n$ constant-degree
algebraic curves or arcs, with a suitable general position assumption. We do not spell out
all the details of this extension, which can be found in a subsequent work of Sharir and Zahl~\cite{ShZ}.
Here we only highlight the necessary modifications.

First, in the case of curves or arcs, one can have cycles of length~$1$ (a curve passing above
itself) or~$2$ (two twisted curves, each passing above the other). This however does not affect
the argument in any significant manner.

The definition of a cycle and of a simple cycle, the $xy$-projection of a cycle, and the path
associated with a cycle, extend to the case of arcs in an immediate and obvious manner.

Guth's polynomial partitioning technique~\cite{Gut} also applies for constant-degree algebraic curves,
with the same performance parameters (albeit with potentially larger constants of proportionality that
depend on the maximum degree of the curves). This allows us to run the same recursive cutting procedure,
and use analogous reasoning to show that it does indeed eliminate all cycles. (The aforementioned
follow-up application of our technique by Sharir and Zahl~\cite{ShZ} spells out the algebraic issues
that arise in such an extension.) It results in a similar
recurrence (with different constants), that solves to the same bound $O(n^{3/2}\log^\beta n)$, albeit
with a larger exponent $\beta$ which now depends on the maximum degree of the arcs.

Finally, to obtain a bound that depends on the number of intersecting pairs of arc projections, we can 
first construct a $(1/r)$-cutting of $\A(\Gamma^*)$, exactly as in the case of segments,\footnote{%
  This extension of the analysis of \cite{dBS}, which had originally been done explicitly only for the case of segments, to more general well-behaved arcs, appears to be considered folklore (see, e.g., \cite{AAS,HP}); refer to the discussion of this issue in the subsequent work~\cite{ShZ}.}
and apply the bound on the number of cuts within each ``pseudo-trapezoid'' of the cutting separately, 
resulting in the following summary result (again, for the detailed and rigorous proof, see~\cite{ShZ}).
\begin{theorem}
  \label{thm:arcs}
  The number of cuts sufficient to eliminate all cycles in a family of $n$ constant-degree 
  algebraic curves or arcs in general position in $\reals^3$ is $O(n + n^{1/2}X^{1/2}\polylog n)$, 
  where $X$~is the number of pairs of arcs whose $xy$-projections cross, and where the constant
  of proportionality and the exponent of the polylogarithmic factor depend on the degree of the input arcs.
\end{theorem}

\begin{remark}
  Both the algorithms of Har-Peled and Sharir\cite{HPS} and of Solan
  \cite{So} rely on a subroutine for quickly checking if a set of
  lines or line segments has an acyclic depth relation.  Analogous
  tools would have to be developed in order to yield an efficient
  construction of a small set of cuts to eliminate all cycles, for the case
  of algebraic arcs and/or curves.

  Concerning the other approach of Aronov et al.~\cite{ABGM}, we do
  not know whether the reduction to Feedback Vertex Cover employed
  there extends to the more general case of algebraic arcs.  If it
  does, it would imply the existence of a polynomial-time
  approximation algorithm for this case as well.
\end{remark}

\end{document}